\documentclass[reprint,aps,prx,showpacs]{revtex4-1}

\pdfoutput=1

\usepackage{graphicx}
\usepackage{dcolumn}
\usepackage{bm}
\usepackage{hyperref}
\usepackage{amsthm}
\usepackage{amsmath}
\usepackage{amssymb}
\usepackage[caption=false]{subfig}

\hypersetup{
	colorlinks = true,
	urlcolor   = blue,
	linkcolor  = blue,
	citecolor  = blue
}

\renewcommand*{\phi}{\varphi}
\renewcommand*{\epsilon}{\varepsilon}
\renewcommand*{\le}{\leqslant}
\renewcommand*{\ge}{\geqslant}

\DeclareMathOperator{\Tr}{Tr}
\DeclareMathOperator{\rank}{rank}

\DeclareMathOperator{\prob}{prob}
\DeclareMathOperator{\poly}{poly}
\DeclareMathOperator{\E}{\mathbb E}
\DeclareMathOperator{\Var}{Var}
\DeclareMathOperator{\erf}{erf}

\newcommand*{\ket}[1]{| #1 \rangle}
\newcommand*{\scalprod}[2]{\langle #1 | #2 \rangle}
\newcommand*{\me}[3]{\langle #1 | #2 | #3 \rangle}
\newcommand*{\dyad}[1]{| #1 \rangle \langle #1 |}

\newcommand*{\HG}{\text{HG}}
\newcommand*{\F}{\mathbb F}

\newtheorem{theorem}{Theorem}

\newcolumntype{2}{D{.}{.}{2}}
\newcolumntype{3}{D{.}{.}{3}}
\newcolumntype{4}{D{.}{.}{4}}
\newcolumntype{5}{D{.}{.}{5}}
\newcolumntype{6}{D{.}{.}{6}}

\begin{document}

\title{Experimental Estimation of Quantum State Properties from Classical Shadows}

\author{G.\,I.\,Struchalin}
	\email{struchalin.gleb@physics.msu.ru}
\author{Ya.\,A.\,Zagorovskii}
\author{E.\,V.\,Kovlakov}
\author{S.\,S.\,Straupe}
\author{S.\,P.\,Kulik}
\affiliation{Quantum Technology Centre, and Faculty of Physics, M.\,V.\,Lomonosov Moscow State University, 119991, Moscow, Russia}

\date{\today}

\begin{abstract}
Full quantum tomography of high-dimensional quantum systems is experimentally infeasible due to the exponential scaling of the number of required measurements on the number of qubits in the system. However, several ideas were proposed recently for predicting the limited number of features for these states, or estimating the expectation values of operators, without the need for full state reconstruction. These ideas go under the general name of shadow tomography. Here we provide an experimental demonstration of property estimation based on classical shadows proposed in [H.-Y.\,Huang \textit{et al.}, Nat. Phys. \href{https://doi.org/10.1038/s41567-020-0932-7}{10.1038/s41567-020-0932-7} (2020)] and study its performance in the quantum optical experiment with high-dimensional spatial states of photons. We show on experimental data how this procedure outperforms conventional state reconstruction in fidelity estimation from a limited number of measurements.
\end{abstract}


\maketitle

\section{Introduction\label{sec:Introduction}}
A full description of a quantum system state is provided by its density matrix~$\rho$, and conventional quantum tomography aims to provide an estimate~$\hat\rho$ of a density matrix for an unknown quantum state given the measurement data~\cite{Paris_Book2004}. The measurements have to be tomographically complete in the sense that they should allow unambiguous determination of all density matrix elements. Simple parameter counting shows that for a general mixed state of a system in a $D$-dimensional Hilbert space, the required number of measurements is at least $D^2$~\cite{Banaszek_PRA99}. This number may be reduced to $O(RD\log^2 D)$ if some prior information about the state rank~$R$ is known by using the techniques of compressed sensing~\cite{Eisert_PRL10}, or otherwise one has to stick to incomplete state tomography~\cite{Teo_QMQM2013}. For pure states, protocols requiring as few as~$5D$ measurements are known~\cite{Delgado_PRL15}. Anyway, for an $n$-qubit system, the number of measurements scales exponentially, since $D=2^n$, which is known as the curse of dimensionality. One of the ways around this problem is to assume some model for the quantum state, allowing for efficient representation, such as a matrix-product state model~\cite{Cramer_NatureComm2010,Lanyon_NaturePhys2017} or a neural-network-based model~\cite{Torlai_NaturePhys2018,Carrasquilla_NatureMI2019,Tiunov_Optica2020}. In general, however, there may be no a priori reason to assign such a model to an unknown state.

On the other hand, the exponential amount of information contained in a full density matrix may be redundant. Typically a researcher is interested in a restricted number of state properties, such as fidelity to the given state which is intended to prepare, or a mean value of some observable. This fact led to a different approach called \emph{shadow tomography} pioneered in the work~\cite{Aaronson_Proc2018}. It promises accurate estimation of exponentially many linear functions of~$\rho$ using only a polynomial number of state copies. However, the original method from Ref.~\cite{Aaronson_Proc2018} is very demanding for hardware implementation as it involves measurements that act collectively on all copies. So despite significant experimental progress in approximate quantum learning~\cite{Rocchetto_ScienceAdvances2019}, direct realization of the original shadow tomography is beyond the current technology.

Fortunately, the authors of Ref.~\cite{Kueng_NatPhys2020} proposed another procedure that requires only separable measurements on each copy yet being powerful in estimating an exponentially large number of state properties. Here we report an experimental realization of this procedure demonstrating estimation of mean values of operators and fidelity estimation from \emph{classical shadows} of quantum states introduced in~\cite{Kueng_NatPhys2020}. We experimentally access Hilbert spaces of dimensionality up to 32 and clearly demonstrate that the approach is applicable in the region of incomplete measurement sets, where traditional tomography fails completely.


\section{Method\label{sec:Method}}
Shadow tomography is a tool for the effective prediction of quantum state properties. Let us note, that understanding the term in this broader sense we will refer to the protocol of Ref.~\cite{Kueng_NatPhys2020} as shadow tomography as well. While it is capable of recovering both linear and higher-order polynomial target functions in matrix elements of~$\rho$, in the present work, we will focus solely on linear ones. We will explicitly describe the algorithm we used in application to our experiment. The reader is referred to the original paper~\cite{Kueng_NatPhys2020} for details on the general framework, nonlinear feature prediction, and proofs of performance guarantees.

The goal of the algorithm is to predict the expectation values $\{ o_i\}$ for a set of~$M$ observables $\{O_i\}$:
\begin{equation}
o_i(\rho) = \Tr O_i \rho, \quad 1 \le i \le M, \label{eq:o}
\end{equation}
where~$\rho$ is an $n$-qubit \emph{true state}. Obviously, $o_i(\rho)$ are linear in matrix elements of~$\rho$.


In the data-gathering stage, $\rho$ is transformed by a unitary operator~$U$, $\rho \to U \rho U^\dagger$, and then each qubit is measured in a computational basis. This procedure is repeated many times for different~$U \in \mathcal U$, chosen randomly from some matrix ensemble~$\mathcal U$. The choice of~$\mathcal U$ affects the tomography performance and determines a class of observables~$O_i$ that can be effectively estimated. The authors of~\cite{Kueng_NatPhys2020} mainly consider two ensembles: stabilizer circuits, i.\,e., $U$ belonging to the $n$-qubit Clifford group~\cite{Gottesman_Thesis}, and Pauli measurements, where each~$U$ is a tensor product of single-qubit operations. We have selected the first option as a more extensive alternative, yet our experimental setup can carry out any measurement.

The random unitary transformation, $\rho \to U \rho U^\dagger$, followed by a measurement in a computational basis~$\{\ket{b_i}\}$ is equivalent to the projection onto a random vector $\ket{\psi_i} = U^\dagger \ket{b_i} \in \mathcal S$. Since~$U$ is a Clifford scheme, then by definition $\ket{\psi}$ is a random \emph{stabilizer} state and~$\mathcal S$ is the set of all $n$-qubit stabilizer states. Later, such measurements will be referred to as Clifford or stabilizer measurements. We resort to vectors, rather than Clifford gates, because our experiment lacks a natural decomposition of unitary transformations into a gate sequence. The algorithm for uniform sampling of random stabilizer states $\ket{\psi} \in \mathcal S$, is presented in the Appendix~\ref{sec:StabilizerGeneration}.

When the measurement results are obtained, the \emph{classical shadow}~$\hat \rho$ of the $n$-qubit true state~$\rho$ is calculated:
\begin{equation}
\hat \rho = (2^n + 1) \sum_{i=1}^{P} f_i \dyad{\psi_i} - \mathbb I, \label{eq:Shadow}
\end{equation}
where~$P$ is the number of projections and $f_i$ is the observed frequency for the outcome, corresponding to~$\ket{\psi_i}$, $\sum_{i=1}^P f_i = 1$. The expression~\eqref{eq:Shadow} is nothing more than an explicit form of a linear inverse (least squares) estimator for any spherical 2-design POVM~\cite{Kueng_JPA2020}. Our choice, i.\,e., stabilizer states, forms a 3-design~\cite{Webb_QIC2016} and the expression is also applicable.

We emphasize that initially, in the work~\cite{Kueng_NatPhys2020}, each projection is assumed to be performed for a single copy of~$\rho$. Therefore, the number of projections~$P$ coincides with the number~$N$ of measured copies, $P = N$, ($f_i = 1/N$). On the other hand, in our quantum optical experiment, several photons can be detected for the same $\ket{\psi_i}$ during the acquisition time, so $P < N$. Moreover, we worked in the regime of overexposure, for which $P \ll N$ (typically, $N/P \sim 10^4$--$10^5$ depending on the system dimensionality). This setting is common in compressive sensing experiments, where shot noise in the outcome probability estimation should be diminished~\cite{Leach_SciRep2014,Eisert_QSciTech2017,Teo_PRA2020}. Preliminary tomography simulations showed that feature prediction accuracy was limited by finite~$P$ even though $N = \infty$ (observed frequency~$f_i$ was substituted with exact outcome probability). In this sense, $P$ is more important than~$N$. When $P < N$, at least, $P$ copies are measured with dissimilar projectors, so theorems presented in Ref.~\cite{Kueng_NatPhys2020} stay valid if~$N$ is replaced by~$P$. However, theorem statements can become pessimistic, and proofs may require further justification for the case $P < N$.

Once a classical shadow~\eqref{eq:Shadow} is obtained, an estimator~$\hat o_i$ of $o_i$ is simply
\begin{equation}
\hat o_i = \Tr O_i \hat \rho, \quad 1 \le i \le M. \label{eq:ohat}
\end{equation}
Here comes another discrepancy with the original algorithm: the authors of Ref.~\cite{Kueng_NatPhys2020} propose to use \emph{median-of-means} estimator. However, we omit the median evaluation and use a simple mean estimator throughout the work, because no valuable difference was found between the two approaches~\footnote{The only exception are the results presented in Fig.~\ref{fig:Medians}.}.

When $P = N$, the shadow tomography protocol has the following sampling complexity~\cite{Kueng_NatPhys2020}:
\begin{theorem}
	$N$ stabilizer measurements suffice to predict $M$ expectations~$o_i = \Tr O_i \rho, 1 \le i \le M,$ within an additive error~$\epsilon$ given that
	\begin{equation}
	N \ge \mathcal O \left( \frac{\log M}{\epsilon^2} \max_i \Tr O_i^2 \right). \label{eq:SampleComplexity}
	\end{equation}
\end{theorem}
The number of copies~$N$ depends on target operators~$O_i$ rather implicitly via $\Tr O_i^2$. In our experiments we used rank-1 projectors, therefore, $\Tr O_i^2 = 1$, and this factor vanishes from~\eqref{eq:SampleComplexity}.

\section{Experiment\label{sec:Experiment}}

We use spatial degrees of freedom of photons to produce high-dimensional quantum states. The corresponding continuous Hilbert space is typically discretized using the basis of transverse modes. We have chosen Hermite-Gaussian (HG) modes $\HG_{nm}(x,y)$, which are the solutions of the Helmholtz equation in Cartesian coordinates $(x,y)$ and form a complete orthonormal basis. The \emph{mode order}~$k$ is defined as a sum of mode indices: $k = n+m$. There exist $(k+1)(k+2)/2$ HG modes from zero to $k$th order inclusive. We bound the beam order to prepare a $D$-dimensional system, i.\,e., the order~$k$ is limited by $k_\text{max}$, $k \le k_\text{max}$, where $k_\text{max}$ is the minimal integer fulfilling the inequality $(k_\text{max}+1)(k_\text{max}+2)/2 \ge D$. We test shadow tomography for dimensions $D = 2, 4, 8, 16,$ and $32$, which corresponds to one to five qubits.

\begin{figure}
	\centering
	\includegraphics[width=\linewidth]{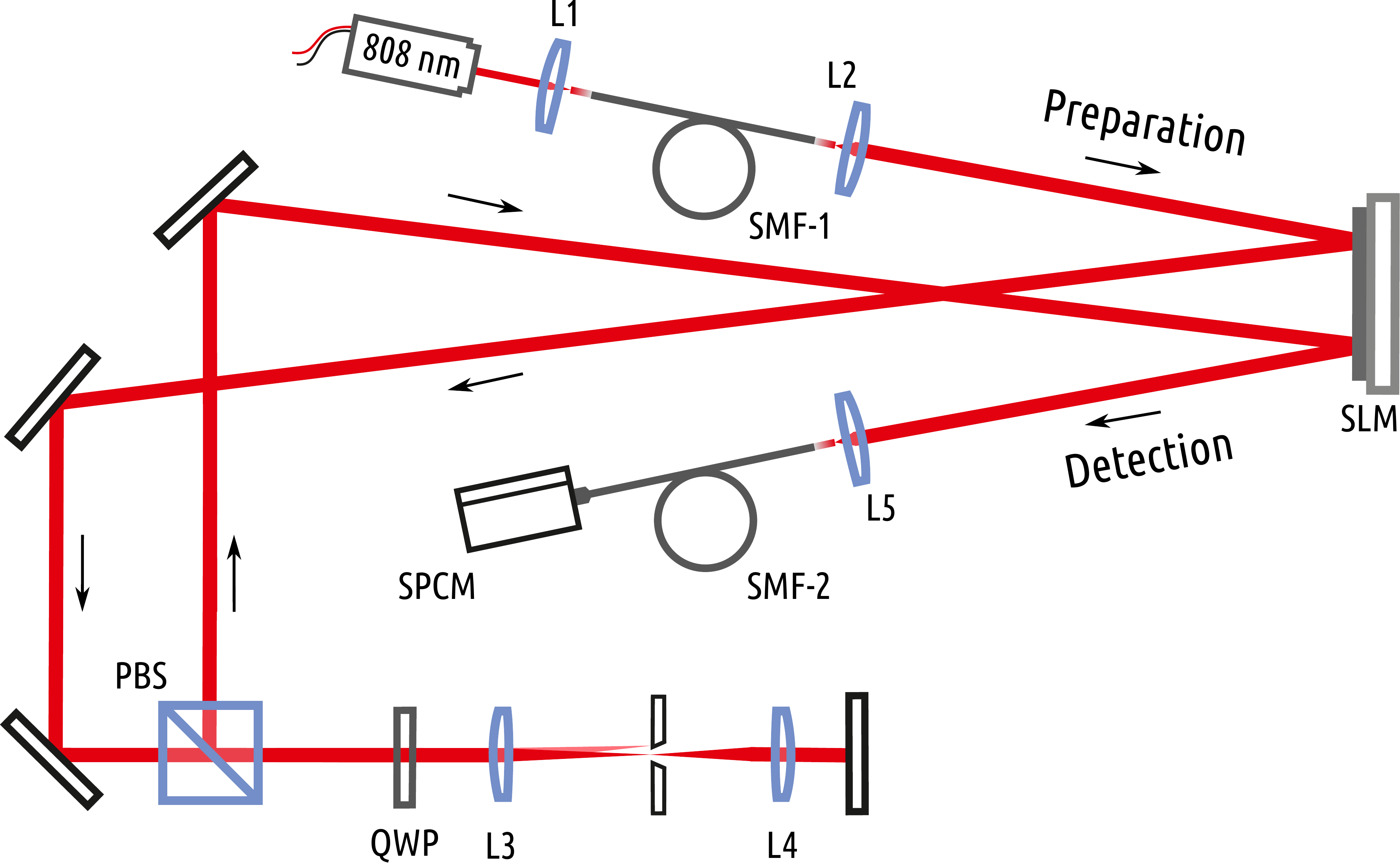}
	\caption{Experimental setup. A spatial light modulator is used for preparation and projective measurements of arbitrary spatial states of photons in a basis of Hermite-Gaussian modes of dimensionality up to 32 (see text for details). }
	\label{fig:Setup}
\end{figure}

In our setup (Fig.~\ref{fig:Setup}) an attenuated light from an 808-nm diode laser is spatially filtered by a single-mode fiber (SMF-1) and collimated by an aspheric lens L2. The top half of a spatial light modulator (SLM, Holoeye Pluto) serves to prepare the desired state of the photon, and the bottom half followed by focusing into a single-mode fiber (SMF-2) and single photon detection implements projective measurements~\cite{Bent:2015experimental, Macarone:2019experimental}. Lenses L3 and L4 have equal focal lengths $F = 100$~mm and are mounted 200~mm apart. Since holograms displayed on the SLM use a blazed grating for amplitude modulation~\cite{Bolduc:2013aa}, the pinhole in the focal plane is used for state selection in the first diffraction order. After a double pass through a telescope and a quarter-wave plate (QWP), the beam is reflected by a polarizing beam splitter (PBS) and directed back to the SLM without any additional alterations.

Note that the detected state differs from the prepared one due to the Gouy phase incursion during beam propagation from one half of the SLM to another. The Gouy phase~$\phi_G$ depends solely on the geometric parameters of the experimental setup, e.\,g., the beam Rayleigh range and traveling distance. It causes the following transformation of basis states: $\ket{\HG_{nm}} \to e^{i (n+m+1) \phi_G}\ket{\HG_{nm}}$. We use  the Gouy phase as a fitting parameter to determine the ``true'' state.

\section{Results\label{sec:Results}}

\subsection{Correlation analysis}

Expectations~$o_i$ can be estimated by shadow tomography via~\eqref{eq:ohat}. On the other hand, the expression~\eqref{eq:o} has the form similar to the Born's rule, so quantities~$o_i$ can be measured directly. It provides a way of independent experimental verification of shadow tomography predictions. We will denote the estimates given by shadow tomography as $\hat o_i^\text{est.}$, and the directly measured expectations as $\hat o_i^\text{meas.}$. These values are both subject to experimental imperfections and shot noise due to finite statistics~$N$. However, the latter factor is negligible, since in all experiments the total exposure corresponding to the value $o_i = 1$ was approximately $3 \times 10^5$ photons with proportional scaling for other values of~$o_i$.

At first, we performed $10^4$ stabilizer measurements to obtain the classical shadow~$\hat \rho$. Then, $5000$ projectors~$O_i = \dyad{\phi_i}$ onto random Haar-distributed vectors~$\ket{\phi_i}$ were measured, resulting in an array of~$\hat o_i^\text{meas.}$. For the same operators~$O_i$, we calculated the predictions $\hat o_i^\text{est.}$ using the classical shadow and plot them against $\hat o_i^\text{meas.}$. For each investigated dimension~$D = 2^n, n = 1, \dots, 5$, we probed five different Haar-distributed random pure true states to ensure that the procedure is a state agnostic one. We observed high Pearson correlation coefficient between the two quantities in all scenarios, signaling about the shadow tomography consistency (see Table~\ref{tab:DrF}).

\begin{table}[h]
	\caption{\label{tab:DrF}Pearson correlation coefficient~$r$ and compensated preparation fidelity~$F$, averaged over five random states, for different system dimensions~$D$.}
	\begin{ruledtabular}
		\begin{tabular}{ccc}
			$D$ & $r$ & $F$ \\
			\hline
			2 & $0.989 \pm 0.002$ & $ 0.981 \pm 0.013$\rule{0pt}{11pt}\\
			4 & $0.983 \pm 0.001$ & $0.974 \pm 0.011$ \\
			8 & $0.976 \pm 0.002$ & $0.899 \pm 0.009$ \\
			16 & $0.953 \pm 0.003$ & $0.920 \pm 0.020$ \\
			32 & $0.875 \pm 0.006$ & $0.807 \pm 0.031$ \\
		\end{tabular}
	\end{ruledtabular}
\end{table}

\begin{figure}
	\centering
	\includegraphics[width=\linewidth]{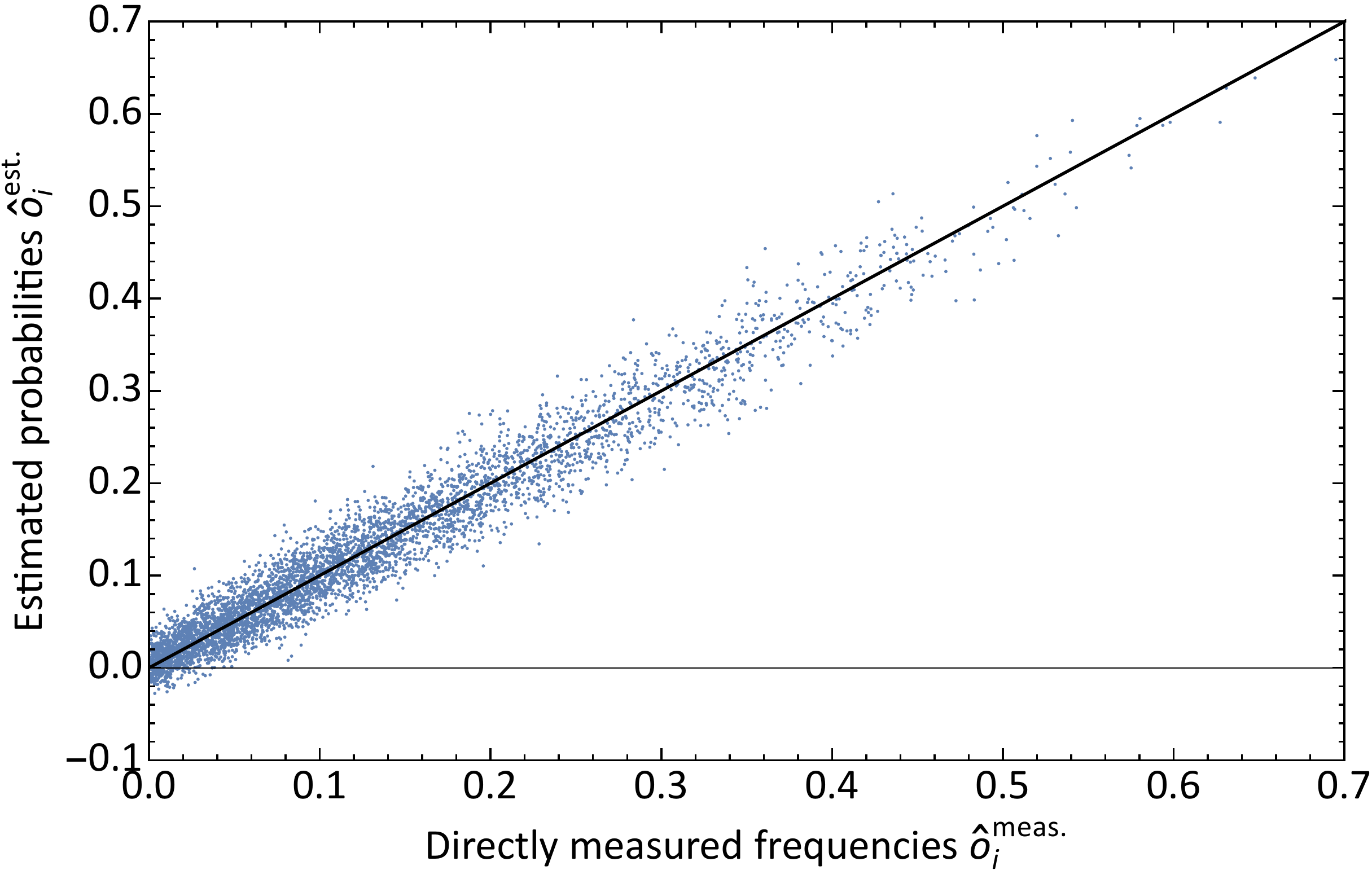}
	\caption{A typical correlation plot (system dimension $D = 8$). Prediction of operator mean values $\hat o_i^\text{est.}$ using shadow tomography versus directly measured quantities $\hat o_i^\text{meas.}$ is depicted. The solid black line corresponds to the equality $\hat o_i^\text{est.} = \hat o_i^\text{meas.}$.}
	\label{fig:ValidationD8Haar}
\end{figure}

A typical correlation plot is depicted in Fig.~\ref{fig:ValidationD8Haar} for system dimension $D = 8$. The solid black line shows perfect matching---the dependence $\hat o_i^\text{est.} = \hat o_i^\text{meas.}$. As one can see all points tend to concentrate near this line (Pearson correlation coefficient is $r = 0.9758$). Note the existence of a small ``nonphysical'' region, where $\hat o_i^\text{est.} < 0$. It appears because the classical shadow~$\hat \rho$ is not forced to be positive semidefinite as in conventional tomography, such as maximum likelihood estimation. And, indeed,~$\hat \rho$ contains negative eigenvalues due to experimental imperfections. Apparently, values of $\hat o_i^\text{meas.}$ are shifted towards zero. This is a mere artifact of our choice for~$O_i$. The probability density function (PDF) for $\hat o_i^\text{meas.}$ coincides with the PDF~$p(x)$ for a squared dot product, $x \equiv |\scalprod{\psi}{\phi}|^2$, between a fixed vector~$\ket{\psi}$, reflecting the true state, and a random Haar-distributed vector $\ket{\phi}$, corresponding to a projector~$O_i$: $p(x) = (D-1)(1-x)^{D-2}$~\cite{Zyczkowski_PRA05}. As the dimensionality~$D$ increases, the mean value $\langle x \rangle = 1/D$ decreases.

\subsection{Effect of median of means estimator}

It was said earlier that the authors of Ref.~\cite{Kueng_NatPhys2020} suggest to use the median of means estimator~\cite{Jerrum_TCS1968}, which proceeds as follows:
\begin{enumerate}
	\item A sequence of~$P$ measurement results is divided into~$K$ batches of length $\lfloor P/K \rfloor$.
	\item An individual shadow~$\hat \rho_k$ is calculated for each batch with index $k = 1, \dots, K$, analogously to~\eqref{eq:Shadow}.
	\item A final assessment $\hat o_i$ is the median:
	\begin{equation}
	\hat o_i = \mathrm{median} (\Tr O_i \rho_1, \dots, \Tr O_i \rho_K).
	\end{equation}
\end{enumerate}
The median of means estimator is robust against outliers in the measured data. The number of batches~$K$ depends on the number of target operators~$M$ and the confidence probability $1-\delta$: $K = 2 \log(2M/\delta)$. For example, if the failure level is chosen to be $\delta = 0.01$ and $M = 5000$, the number of batches is $K \approx 28$.

\begin{figure}[b]
	\centering
	\includegraphics[width=\linewidth]{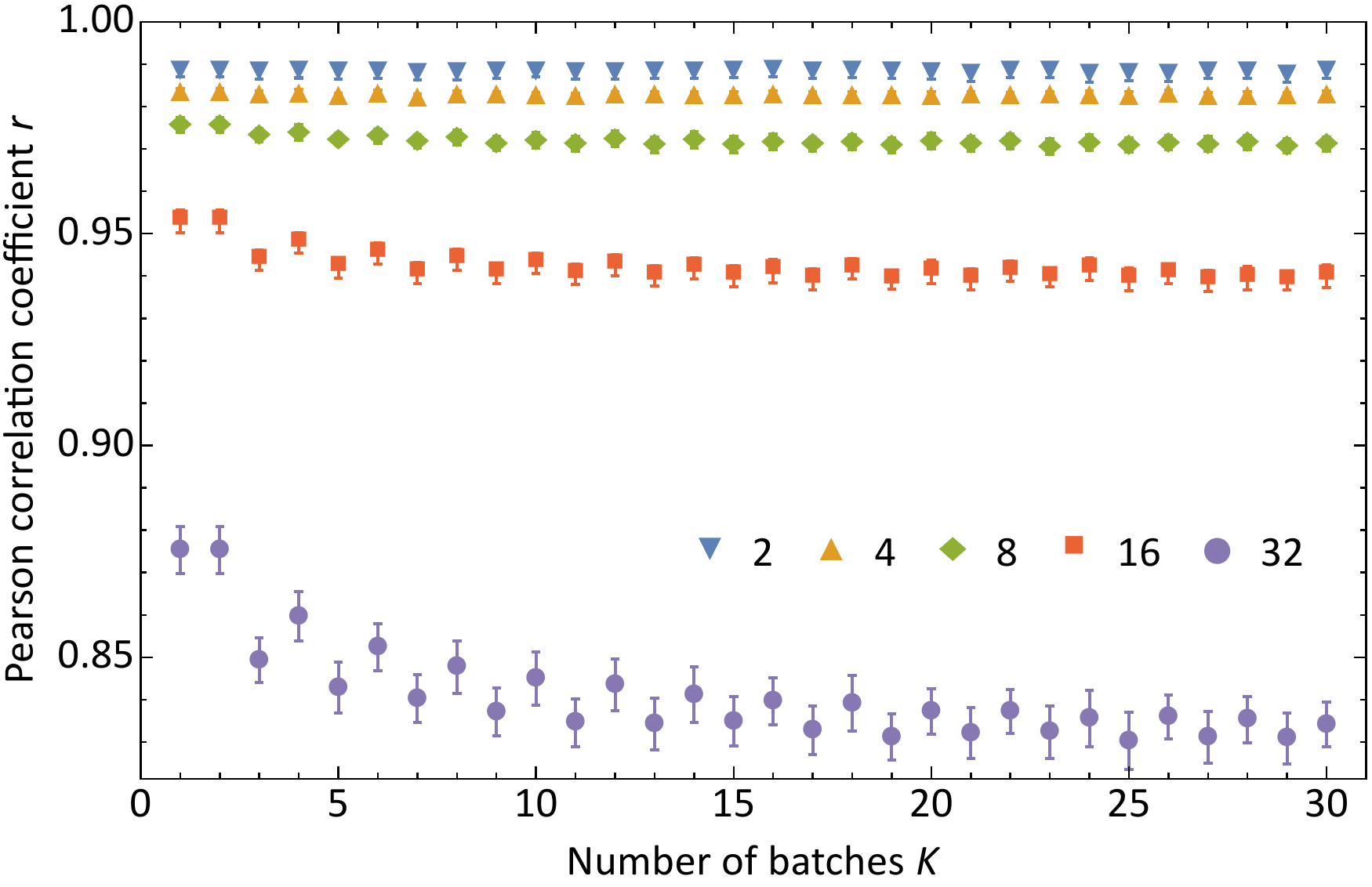}
	\caption{Dependence of the Pearson correlation coefficient~$r$ between $\hat o_i^\text{est.}$ and $\hat o_i^\text{meas.}$ on the number of batches~$K$ in the median of means evaluation for different system dimensions~$D$ (see legends). Each data point is averaged over five true states. Error bars correspond to one standard deviation of the mean.}
	\label{fig:Medians}
\end{figure}

In order to investigate how the number of batches~$K$ influences the overall tomography performance, we found the median-of-means predictions $\hat o_i^\text{est.}$ for various~$K$ and calculated the Pearson correlation coefficient~$r$ between $\hat o_i^\text{est.}$ and $\hat o_i^\text{meas.}$. The obtained dependencies $r(K)$ are presented in Fig.~\ref{fig:Medians} for different system dimensions~$D$. Each curve is averaged over five tomography runs. The case $K = 1$ corresponds to the ordinary mean estimator, as was used before. The reader can see that the dependencies $r(K)$ are almost flat, and correlation even becomes slightly lower with the increase of~$K$. This implies that in application to our experiment the effect of the median of means estimator is negligible compared to the mean alone.

We connect the independence of accuracy on~$K$ with two facts. Firstly, the statistics per measurement in our experiments is huge, and the outliers hardly occur. See Appendix~\ref{sec:MedianOfMeans} for more detailed reasoning. Secondly, systematic, deterministic errors in measurement projectors dominate over the statistical noise, and medians cannot smooth away this source of imperfections.

\subsection{Fidelity estimation}

One of the important usecases for shadow tomography is the estimation of fidelity to some given pure state~$\ket{\psi}$. In this case, the target operator~$O$ is simply a projector onto this state: $O = \dyad{\psi}$. In particular, one can find fidelity of the state preparation. However in our experiment, the prepared state~$\ket{\psi_\text{prep.}}$ and the detected one~$\ket{\psi_\text{det.}}$ differ significantly due to the Gouy phase incursion during the beam propagation, and we have to perform the corresponding correction (see Appendix~\ref{sec:GouyPhase} for details). The obtained fidelities~$F$ are listed in Table~\ref{tab:DrF}.

The results presented above were obtained using overcomplete measurement sets since we used $P = 10^4$ projectors to construct the classical shadow~$\hat \rho$. This number is far greater than the size of a minimal complete set, which has $D^2 - 1$ POVM elements, even for $D = 32$. However, the main distinguishing feature of shadow tomography is its ability to predict expectation values using much less then a tomographically complete set of measurements. Hence, we also studied the performance of shadow tomography for the intermediate values of~$P$, including the incomplete scenario, where $P < D^2 - 1$.

\begin{figure*}
	\centering
	\subfloat[Shadow tomography.]
	{
		\includegraphics[width=0.45\linewidth]{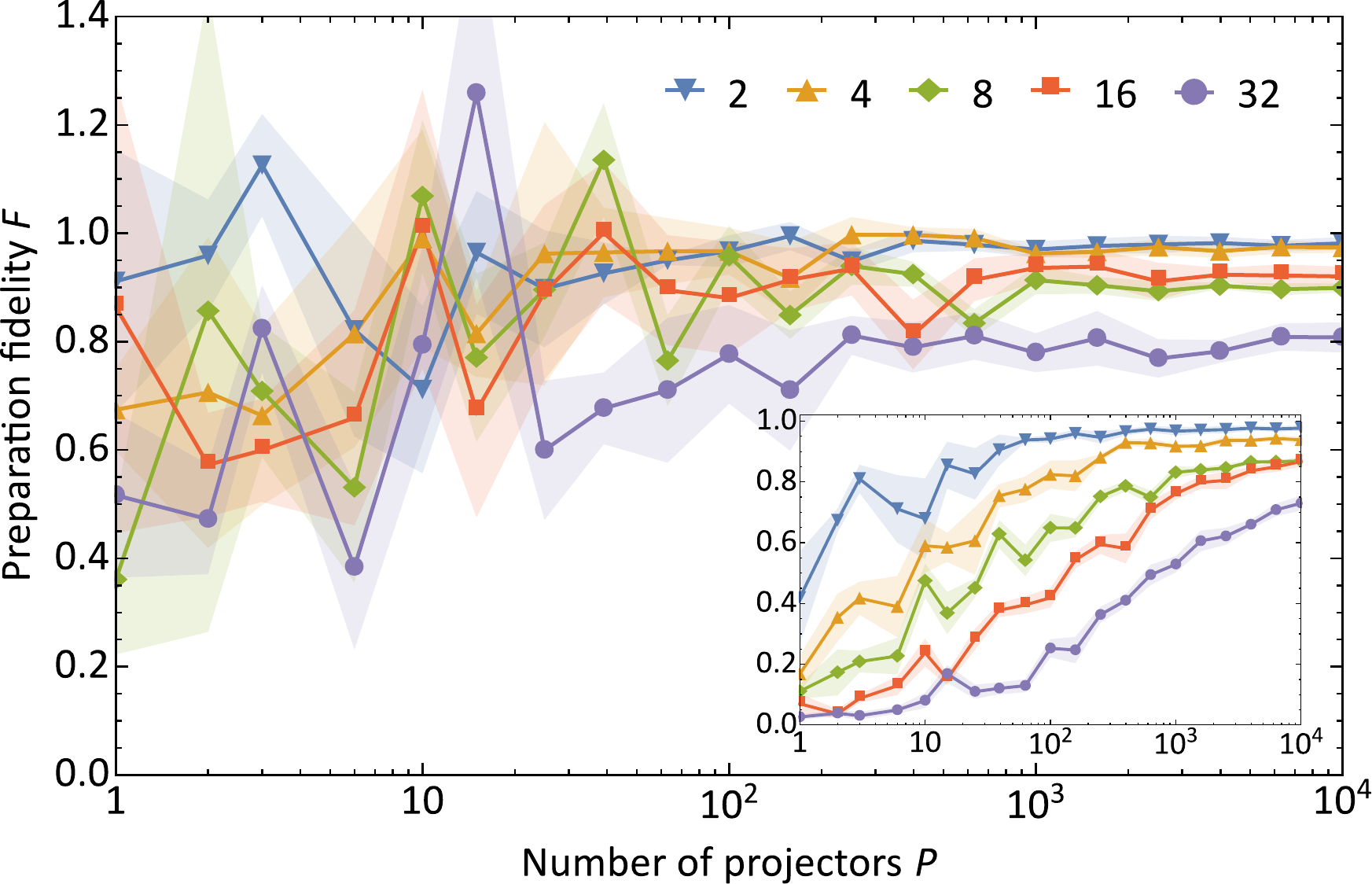}
		\label{fig:FidelityST}
	}
	\quad
	\subfloat[Maximum likelihood estimation.]
	{
		\includegraphics[width=0.45\linewidth]{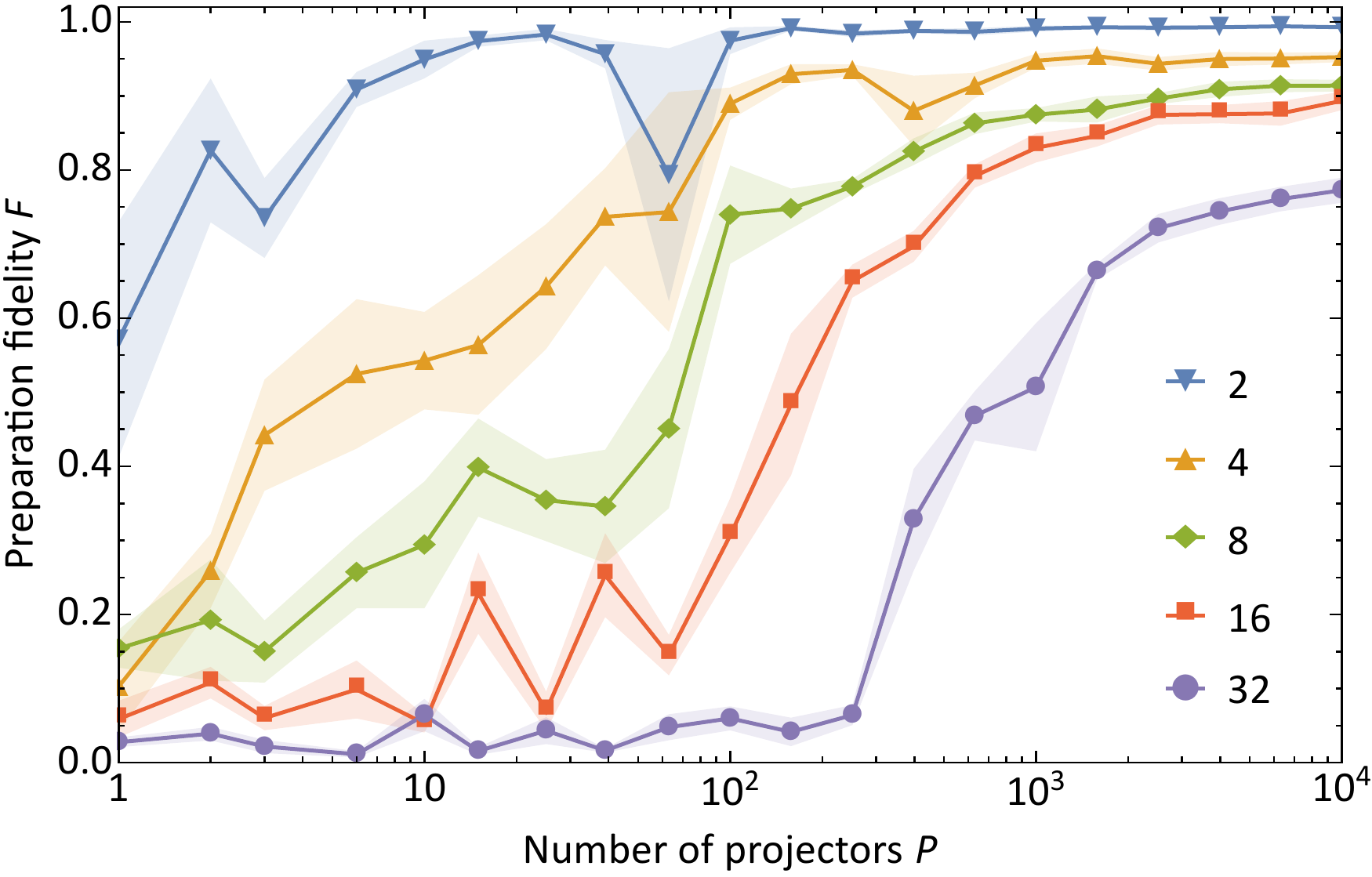}
		\label{fig:FidelityMLE}
	}
	\caption{Compensated preparation fidelity~$F$ on the number of stabilizer measurements~$P$ for different system dimensions~$D$ (see legends) obtained using (a) shadow tomography and (b) maximum likelihood estimation. Each curve is averaged over five true states. Shaded area corresponds to one standard deviation of the mean. Inset of Fig.~\ref{fig:FidelityST} shows the same dependencies, but the classical shadow~$\hat \rho$ is projected onto the set of physical density matrices.}
	\label{fig:Fidelity}
\end{figure*}

Fig.~\ref{fig:FidelityST} shows averaged dependencies of the preparation fidelity~$F$, estimated using shadow tomography, on the number of stabilizer measurements~$P$ for various system dimensions~$D$. Fidelity is calculated with respect to the compensated prepared state, where the compensatory Gouy phase is found using the full data sequence (i.\,e., for $P = 10^4$). The averaging is done over five different states for each dimension.

In the beginning, for low~$P$, the volatility of curves is vast, and fidelity~$F$ can even lie outside the physical region $0 \le F \le 1$ due to the negative definiteness of a shadow matrix~$\hat \rho$. As~$P$ increases, fidelities start to stabilize near their final values. Nevertheless, the fidelity estimators are unbiased for any number of projectors~$P$ because shadow tomography is based on the linear inversion that is unbiased. And indeed, as one can see from Fig.~\ref{fig:FidelityST}, the error bars cover the final values of fidelity reasonably well for any~$P$, which experimentally confirms the unbiasedness property.

It is interesting to see how the above fidelity estimates change if the shadow matrix~$\hat \rho$ [see Eq.~\eqref{eq:Shadow}] is forced to be positive semidefinite. To achieve this, we project the eigenvalues~$\lambda_i$ of~$\hat \rho$ onto a canonical simplex $\Delta = \{(\lambda_1, \dots, \lambda_D) \mid \lambda_i \ge 0 \wedge \sum_{i=1}^D \lambda_i = 1\}$, using the recipe from Ref.~\cite{Chen_Arxiv2011}, while leaving the eigenvectors untouched. The obtained results are shown in the inset of Fig.~\ref{fig:FidelityST}. Now the estimators are biased: for incomplete measurement sets, $P \lesssim D^2$, fidelity is underestimated and significantly shifted towards zero. When~$P$ becomes equal in the order of magnitude to $D^2$, the assessments attain their final values. Note the apparent dependency on the system dimension~$D$, which is not the case for ordinary shadow tomography.

The bias of the estimator leads to poor accuracy when a measurement set is incomplete. For example, consider the point with $D = 32$ and $P = 251$ in Fig.~\ref{fig:FidelityST}. Shadow tomography has already converged since fidelity is $F = 0.81 \pm 0.04$, which coincides with the final value for $P = 10^4$ within the error bars, but after the projection of eigenvalues onto the positive simplex fidelity drops to $F = 0.36 \pm 0.03$. Unfortunately, the bias is unavoidable for any procedure that always yields positive density matrices~\cite{Schwemmer_PRL2015}. A maximum likelihood estimate (MLE) is neither an exclusion. Fig.~\ref{fig:FidelityMLE} presents fidelity dependencies for the same measurement data, processed with an accelerated projective gradient MLE~\cite{Shang_PRA2017}. Qualitatively, the performance is the same as the one for the inset of Fig.~\ref{fig:FidelityST}.

\subsection{Estimator biasedness}

\begin{figure*}
	\centering
	\subfloat[$D = 8, P = 100$.]
	{
		\includegraphics[height=35mm]{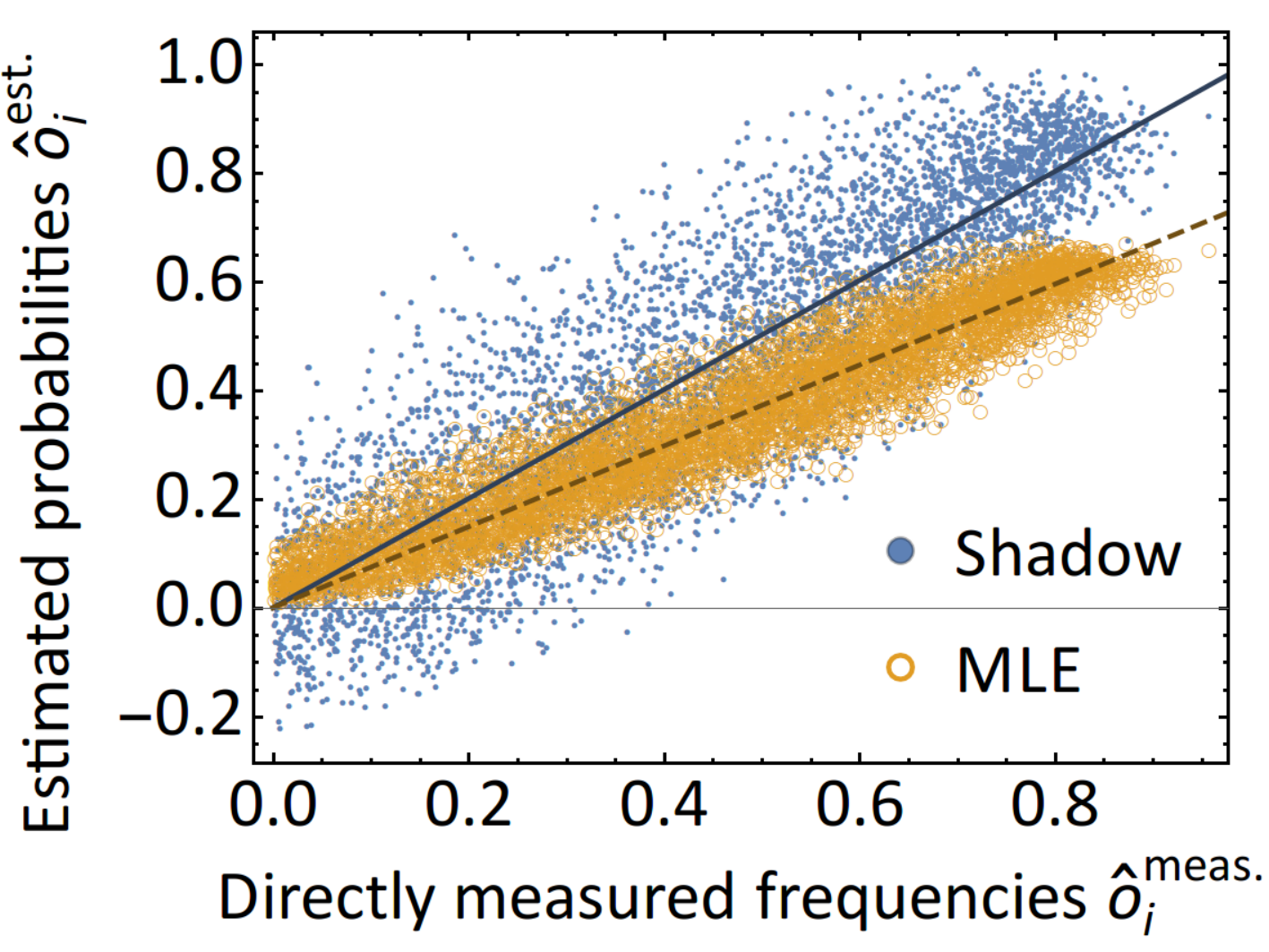}
		\label{fig:ValidationD8PLow}
	}
	\subfloat[$D = 8, P = 10^4$.]
	{
		\includegraphics[height=35mm]{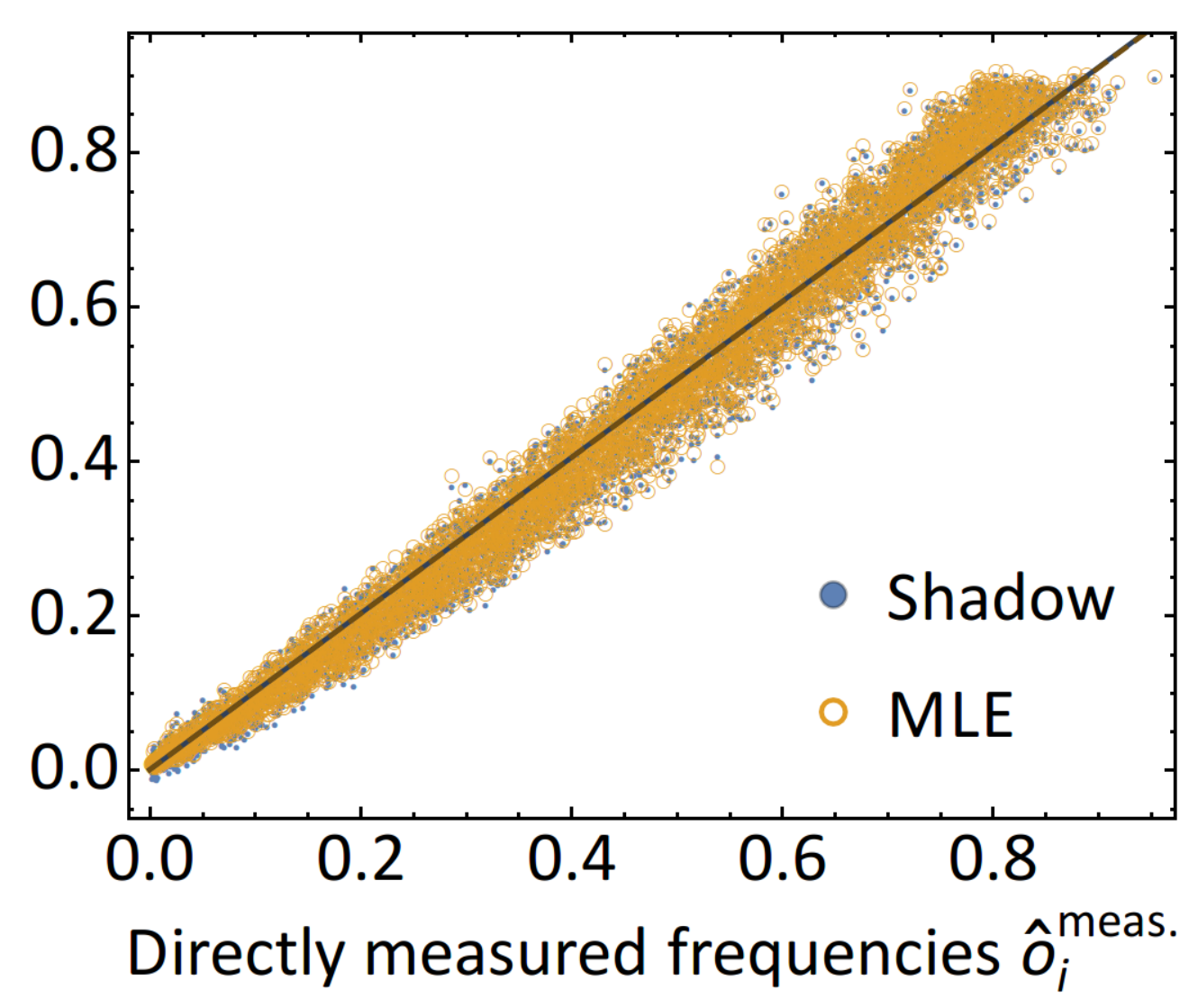}
		\label{fig:ValidationD8PHigh}
	}
	\subfloat[$D = 32, P = 300$.]
	{
		\includegraphics[height=35mm]{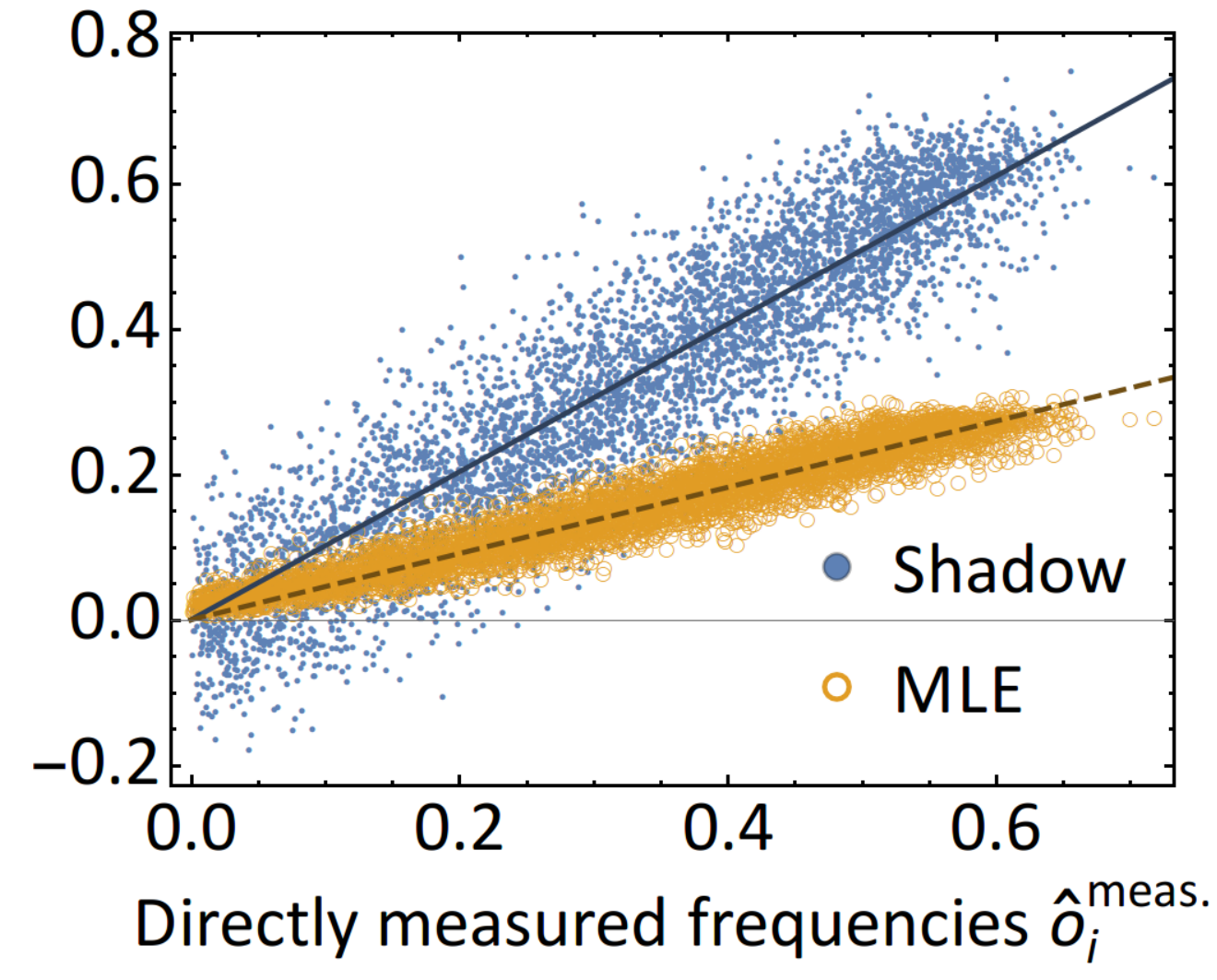}
		\label{fig:ValidationD32PLow}
	}
	\subfloat[$D = 32, P = 10^4$.]
	{
		\includegraphics[height=35mm]{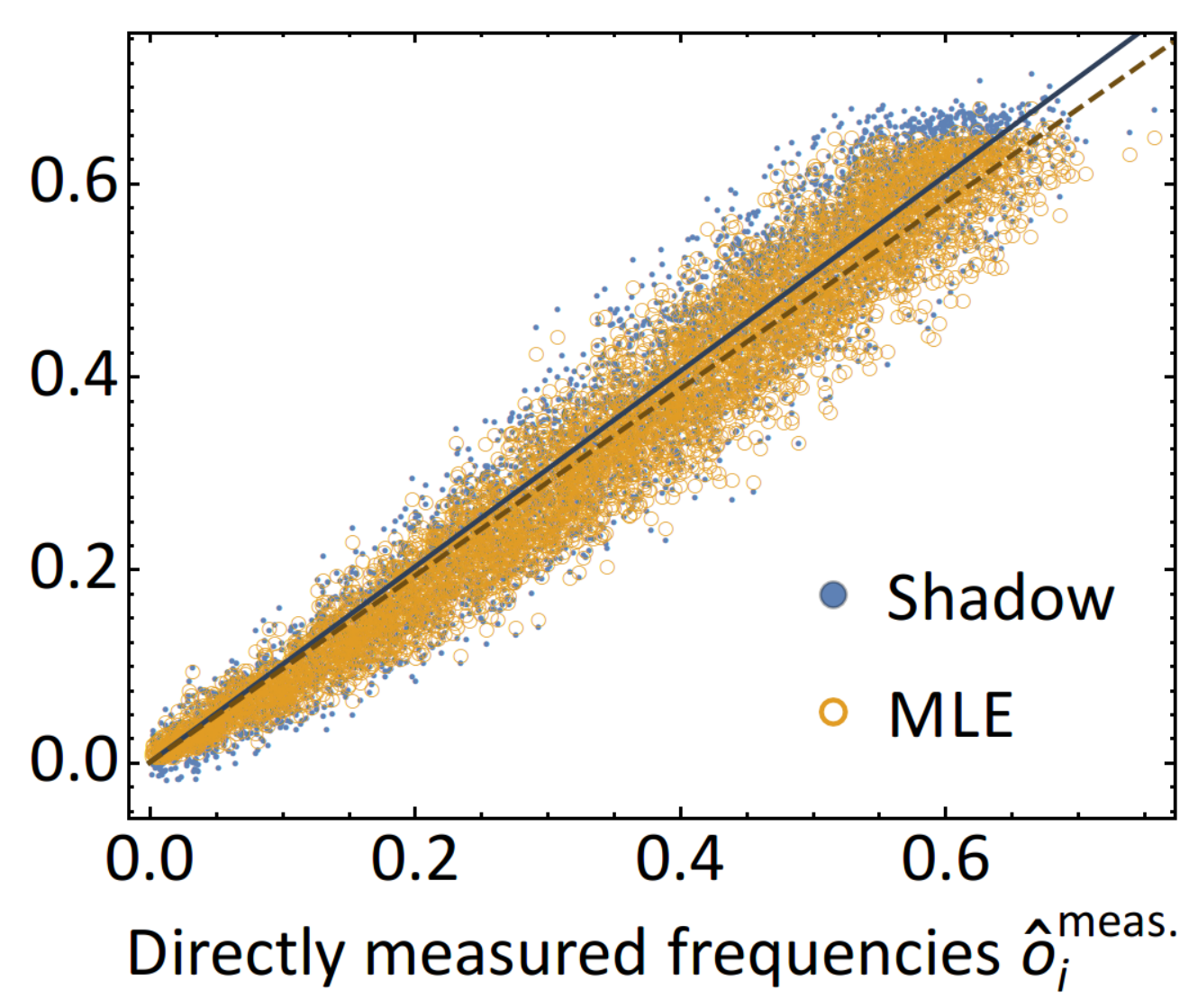}
		\label{fig:ValidationD32PHigh}
	}
	\caption{Comparison of correlation plots obtained using  shadow tomography and maximum likelihood estimation (MLE) for different system dimensions~$D$ and number of stabilizer measurements~$P$. Prediction of operator mean values $\hat o_i^\text{est.}$ using tomographic methods versus directly measured quantities $\hat o_i^\text{meas.}$ is depicted. Straight lines are best-fit dependencies of the form $\hat o_i^\text{est.} = \beta \hat o_i^\text{meas.}$ (solid lines---shadow tomography, dashed lines---MLE). For overcomplete number of measurements $P = 10^4$ (Fig.~\ref{fig:ValidationD8PHigh} and~\ref{fig:ValidationD32PHigh}) both methods result in the same unbiased predictions with a proportionality coefficient $\beta \approx 1$. For low values of~$P$ (Fig.~\ref{fig:ValidationD8PLow} and~\ref{fig:ValidationD32PLow}) MLE predictions are highly biased and underestimate~$\hat o_i^\text{meas.}$, while classical shadow assessments are still unbiased.}
	\label{fig:Validation}
\end{figure*}

Preparation fidelity is not the only quantity estimated with heavy bias employing the MLE method when the number~$P$ of stabilizer measurements is low. All projectors~$O_i$ with the near-unity mean value~$o_i \approx 1$ will be underestimated. To check this hypothesis, we carried out another correlation-like test, similar to those in Fig.~\ref{fig:ValidationD8Haar}. Both a classical shadow~$\hat \rho_\text{CS}$ and a maximum likelihood estimate~$\hat \rho_\text{MLE}$ are calculated using the same stabilizer measurements outcomes. Then as usual, these estimators are substituted into Eq.~\eqref{eq:ohat} to give~$\hat o_i^\text{est.}$ for a set of 5000 randomly chosen projectors~$O_i$.

We note that the difference between shadow and MLE tomography is visible the most in the region, where~$o_i \approx 1$. At the same time, Haar-distributed projectors~$O_i$ tend to have low mean values $o_i$ (on average $\langle o_i \rangle = 1/D$), which do not suit well for this kind of test. Therefore, we select random projectors~$O_i$ with uniformly distributed expectations $o_i$. To do so, they should be adjusted to the true state. In particular, we use projectors~$O_i = \dyad{\phi}$ onto a random vector~$\ket{\phi}$:
\begin{equation}
\ket{\phi} = \sqrt{a} \ket{\psi} + \sqrt{1 - a} \frac{\ket{g} - \ket{\psi} \scalprod{\psi}{g}}{\| \ket{g} - \ket{\psi} \scalprod{\psi}{g} \|}, \label{eq:UniprobMeas}
\end{equation}
where~$\ket{g}$ is a vector with real and imaginary parts of its elements being independent Gaussian random variables with zero mean and unit variance and~$a$ is distributed uniformly on the interval $[0, 1]$. It is easy to verify that $|\scalprod{\psi}{\phi}|^2 = a$, so if $\ket{\psi}$ is the true state, then, indeed, $o_i = a$ has uniform distribution. We take a close approximation---a compensated prepared state---as the vector~$\ket{\psi}$. The choice of distribution for $\ket{g}$ ensures that a ``circle'' determined by the equation $|\scalprod{\psi}{\phi}|^2 = \text{const}$ is also populated uniformly~\cite[Appendix C]{Kulik_PRA16}.

Obtained predictions~$\hat o_i^\text{est.}$ against directly measured mean values~$\hat o_i^\text{meas.}$ are shown in Fig.~\ref{fig:Validation} for system dimensions $D = 8$ and $32$. We investigated two cases: estimates for small number of measurements~$P$ (100 for $D = 8$ and 300 for $D = 32$) and large $P = 10^4$. As expected, shadow tomography gives unbiased estimates in all situations: $\hat o_i^\text{est.} \approx \hat o_i^\text{meas.}$. MLE performs differently, since it is only an \emph{asymptotically} unbiased estimator. For small~$P$, although the predictions are more condensed compared to classical shadows (there is less volatility), they are underestimated and concentrate near a line $\hat o_i^\text{est.} = \beta \hat o_i^\text{meas.}$ with proportionality constant $\beta < 1$ (see Table~\ref{tab:ValidationFit} for best-fit parameters). For large~$P$ the behavior equalizes: MLE approaches the asymptotic and produces unbiased estimates that almost coincide with those calculated using classical shadows. We connect the observed flat-top cutoff under $o_i^\text{est.} = 1$ in Figs.~\ref{fig:ValidationD8PHigh} and~\ref{fig:ValidationD32PHigh} with that our choice of $\ket{\psi}$ in Eq.~\eqref{eq:UniprobMeas} differs from the true state $\rho$.

\begin{table}[h]
	\caption{\label{tab:ValidationFit}Pearson correlation coefficient~$r$ and proportionality coefficient~$\beta$ of the data in Fig.~\ref{fig:Validation} obtained using classical shadows (CS) and maximum likelihood estimation (MLE) for different system dimensions~$D$ and number of stabilizer measurements~$P$.}
	\begin{ruledtabular}
		\begin{tabular}{cccccc}
			$D$ & $P$ & $r_\text{CS}$ & $r_\text{MLE}$ & $\beta_\text{CS}$ & $\beta_\text{MLE}$ \\
			\hline

			8 & 100 & 0.870 & 0.949 & $1.004 \pm 0.004$ & $0.745 \pm 0.002$\rule{0pt}{11pt}\\

			8 & $10^4$ & 0.990 & 0.990 & $1.011 \pm 0.001$ & $1.010 \pm 0.001$ \\
			32 & 300 & 0.915 & 0.957 & $1.016 \pm 0.003$ & $0.455 \pm 0.001$ \\
			32 & $10^4$ & 0.971 & 0.974 & $1.013 \pm 0.002$ & $0.967 \pm 0.002$ \\
		\end{tabular}
	\end{ruledtabular}
\end{table}

\section{Conclusion\label{sec:Conclusion}}

We have experimentally demonstrated that classical shadows, i.\,e., linear inversion estimators for quantum states can be used to faithfully predict expectation values of observables from very few measurements. Specifically, we have shown that the estimator obtained from the classical shadow is unbiased and provides correct expectation values even when the number of measurements used for estimation is significantly less than required for full state reconstruction. As a special case we performed estimation of fidelity with the ``true'' state and shown that it is also possible with few measurements.

Our treatment reformulates the results of~\cite{Kueng_NatPhys2020} in terms of a typical quantum optical experiment and is then applied to experimental data for high-dimensional spatial states of photons. The versatility of the chosen experimental platform allows us to realize arbitrary projective measurements; however, we have demonstrated that in full accordance with the theoretical predictions, the procedure works well when the measurement set is restricted, for example, to projections on the stabilizer states. This is an important feature of the protocol, making it a scalable approach to quantum property estimation.

The framework of shadow tomography was recently extended with online learning protocols~\cite{Aaronson_JSM2019,Chen_Arxiv2020}, which from an operational point of view are close in spirit to the one implemented in this work. Comparing the performance of these approaches on real experimental data is an interesting direction for further research.

\begin{acknowledgments}
We acknowledge financial support from the Russian Foundation for Basic Research (RFBR Project No. 19-32-80043 and RFBR Project No. 19-52-80034) and support under the Russian National Technological Initiative via MSU Quantum Technology Centre.
\end{acknowledgments}	

\appendix

\section{Explicit procedure for generation of random stabilizer states\label{sec:StabilizerGeneration}}
In this section, we describe the procedure for the explicit generation of random, uniformly distributed, stabilizer states. By ``explicit'' we mean that the whole $n$-qubit state vector~$\ket{\psi}$ of $2^n$ amplitudes is calculated. This requirement comes from the fact that in photonic experiments like the one performed here the preparation and measurement stage has no natural decomposition in terms of quantum gates and requires the explicit specification of the state vectors.

The set~$\mathcal S$ of all stabilizer states is finite, its cardinality~$C(n)$ is~\cite{Gottesman_PRA2004}:
\begin{equation}
C(n) = 2^n \prod_{k = 1}^{n} (2^k +1) \approx 2^{n^2/2}. \label{eq:C(n)}
\end{equation}
Uniform sampling means that each state $\ket{\psi_i} \in \mathcal S$ is selected with equal probability. A naive approach would be to generate a random index $i = 1, \dots, C(n)$, and pick the corresponding state~$\ket{\psi_i}$ from a pre-generated set~$\mathcal S$. However, the huge cardinality makes it infeasible.

When working with stabilizer states on a classical computer, one usually resorts to their stabilizer operators rather than vectors, since this implicit description allows very efficient (polynomial in the number of qubits~$n$) storage scheme and simulation of Clifford gate actions. This fact is known as the Gottesman--Knill theorem~\cite{Gottesman_Proc1998,Gottesman_PRA2004}.

Therefore, an evident practical approach for constructing a random~$\ket{\psi}$ is to generate a set of its stabilizers~$\{g_i\}_{i=1}^n$: $g_i \ket{\psi} = \ket{\psi}$. This can be done efficiently by utilizing, e.\,g., a method from Ref.~\cite{Koenig_JMP2014}, which enumerates all possible stabilizers circuits~$U$: $\ket{\psi} = U \ket{0}$. Then given the stabilizers $\{g_i\}$ the state $\ket{\psi}$ is obtained using the relation:
\begin{equation}
\dyad{\psi} = \prod_{i = 1}^{n} \frac{1+g_i}{2}. \label{eq:StabToState}
\end{equation}
Thus, the conversion from a stabilizer formalism to an explicit form involves three exponentially hard routines:
\begin{enumerate}
	\item an explicit construction of stabilizer matrices $g_i$---$\mathcal O(n \cdot 2^{2n})$ operations,
	\item product evaluation---$\mathcal O(n \cdot 2^{3n})$,
	\item recovering of $\ket{\psi}$ from $\dyad{\psi}$---$\mathcal O(2^n)$.
\end{enumerate}
The overall complexity is dominated by the second stage (note the power $3n$).

Of course, the complexity of an explicit $n$-qubit state generation is always exponential and cannot be lower than $O(2^n)$---the number of operations required to address every element in the vector. But the power index dramatically affects the performance. For the method above, it is $3n$, while a decrease to $n$ is possible. Below, we describe an algorithm that requires $\mathcal O(2^n\poly(n))$ operations.

Let us start with the universal form of any~$\ket{\psi} \in \mathcal S$~\cite{Dehaene_PRA2003, Nest_QIC2010}:
\begin{equation}
\ket{\psi} \propto \sum_{x \in \F_2^k} (-1)^{q(x)} i^{l(x)} \ket{Rx + t}, \label{eq:GenStabForm}
\end{equation}
where $x \in \F_2^k$ and $t \in \F_2^n$ are, respectively, $k$- and $n$-dimensional binary vectors, $k \le n$, $q(x)$ is a quadratic form on $\F_2^k$, $l(x)$ is a linear one, and $R \in \F_2^{n \times k}$ is an $n \times k$ binary matrix with rank~$k$. Summation and multiplication in~\eqref{eq:GenStabForm} are carried modulo two, since we work in a Galois field~$\F_2$. Also, we identify a binary representation of a given integer number $x$ with the corresponding binary vector and vice versa.

Representation~\eqref{eq:GenStabForm} reveals some properties of stabilizer states. Up to normalization, each element of the state can be either $\pm 1$, $\pm i$, or 0. The number of nonzero elements is always $2^k$, $0 \le k \le n$, which is simply the number of different vectors $x$ in $\F_2^k$. By convention, we define $\F_2^0 = \{ 0 \}$.

In our sampling algorithm the set~$\mathcal S_k$, which by definition contains all $n$-qubit stabilizer states with $2^k$ nonzero elements, plays an important role. In Theorem~\ref{thm:SampleSk} we show how to sample $\ket{\psi} \in \mathcal S_k$ uniformly. Then in Theorem~\ref{thm:SkCardinality} we calculate the cardinality $C(n, k)$ of~$\mathcal S_k$. Finally, we combine these results in Theorem~\ref{thm:SampleS}, where an algorithm for uniform sampling of the whole set $\mathcal S = \bigcup_{k = 0}^n \mathcal S_k$ is presented.

\begin{theorem}\label{thm:SampleSk}
Fix $k = 0, \dots, n$. Let~$\ket{\psi}$ be an $n$-qubit state of the form
\begin{align}
\ket{\psi} &:= \ket{t}, \text{ if } k = 0, \nonumber \\
\ket{\psi} &:= \frac{1}{2^{k/2}} \sum_{x=0}^{2^k - 1} (-1)^{x^T Q x} i^{c^T x} \ket{Rx + t}, \text{ if } k \ne 0, \label{eq:GenStabForm2}
\end{align}
where $x \in \F_2^k$. Quantities $Q \in \F_2^{k \times k}$, $c \in \F_2^k$, $t \in \F_2^n$ are random with independent and identically distributed (i.\,i.\,d) elements 0 or 1 appearing with probability $1/2$. $R \in \F_2^{n \times k}$ ($\rank R = k$) is a random matrix sampled uniformly from the set of all rank-$k$ matrices. Then $\ket{\psi}$ is uniformly sampled from $\mathcal S_k$.
\end{theorem}

\begin{proof}
Again, consider~\eqref{eq:GenStabForm}. Every quadratic form~$q(x)$ can be expressed as a sum: $q(x) = x^T Q x + b^T x + x_0$. The constant~$x_0$ affects only the global phase of~$\ket{\psi}$ and can be omitted. The linear term $b^T x$ is already enclosed in $x^T Q x$. Indeed, in the expansion of $x^T Q x$, there is a diagonal term $Q_{ii} x_i^2 = Q_{ii} x_i$, since $x_i^2 = x_i$ for any $x_i \in \F_2$. Therefore, without loss of generality, $q(x) = x^T Q x$, where $Q$ is an arbitrary $k \times k$ binary matrix. Analogously, a constant term can be neglected in the linear form: $l(x) = c^T x + x_0 \sim c^T x$, where $c \in \F_2^k$ is an arbitrary binary vector. Started from the form~\eqref{eq:GenStabForm}, we have already arrived at a more specific expression~\eqref{eq:GenStabForm2}.

Let us prove that $\ket{\psi}$ is sampled uniformly. Quantities $Q, c, R, t$ are associated with their own structures, respectively, a quadratic form~$\mathcal Q$, a linear form $L$, a $k$-dimensional vector subspace~$V_k$, and an affine subspace $A_k$:
\begin{gather}
\mathcal Q = f_Q(Q) = \{(x, x^T Q x) \mid x \in \F_2^k\}, \label{eq:f_Q} \\
L = f_c(c) = \{(x, c^T x) \mid x \in \F_2^k\},  \label{eq:f_c} \\
V_k = f_R(R) = \{Rx \mid x \in \F_2^k\},  \label{eq:f_R} \\
A_k = f_t(t) = \{y + t \mid y \in V_k\}. \label{eq:f_t}
\end{gather}
The corresponding maps $f_Q, f_c, f_R, f_t$ are in general surjective, i.\,e., they may map many different quantities to a single structure. An affine subspace~$A_k$ determines positions of nonzero elements in~$\ket{\psi}$, while forms~$\mathcal Q$ and~$L$ define the order in which $\pm 1, \pm i$ appear. In this sense, $\mathcal Q, L$, and~$A_k$ act independently, so the state~\eqref{eq:GenStabForm2} is uniformly distributed if each of these three structures is sampled uniformly.

Obviously, $Q, c$, and $t$ are generated uniformly, because their elements are i.\,i.\,d. random variables with $\prob(0) = \prob(1) = 1/2$; $R$ is sampled uniformly as the theorem condition states. However, for a general surjective map, $\omega = f(\xi), \xi \in \Xi, \omega \in \Omega$, a uniform sampling of the domain~$\Xi$ does not imply the same for its image~$\Omega$. Fortunately, for the maps~\eqref{eq:f_Q}--\eqref{eq:f_t} cardinality of a preimage for each element in a codomain is the same (see below):
\begin{equation}
|f^{-1}(\omega_1)| = |f^{-1}(\omega_2)|, \quad \forall\, \omega_{1,2} \in \Omega, \label{eq:PreimageCardinality}
\end{equation}
where~$|\cdot|$ denotes the cardinality evaluation. Any map~$f$ satisfying~\eqref{eq:PreimageCardinality} has the property that a uniform sampling of the domain~$\Xi$ results in a uniform sampling of its codomain~$\Omega$. 

Let us start with proving that the map~$f_Q$ complies~\eqref{eq:PreimageCardinality}. Only the sum $(Q_{ij} + Q_{ji})x_i x_j, i < j,$ matters in the expression $x^T Q x$. So adjoint nondiagonal elements~$Q_{ij}$ and~$Q_{ji}$ can be replaced by their equivalents: a pair $(Q_{ij}, Q_{ji}) = (0, 1)$ is equivalent to $(1, 0)$, and $(0, 0) \sim (1, 1)$. There are $k(k-1)/2$ such pairs in a $k \times k$ matrix~$Q$ and each pair has two equivalent values. Therefore, for every~$\mathcal Q$: $|f^{-1}_Q(\mathcal Q)| = 2^{k(k-1)/2}$.

There is a one-to-one correspondence between vectors~$c$ and linear forms~$L$, so the map~$f_c$ is bijective, and $|f^{-1}_c(L)| = 1$.

Any nondegenerate $n \times k$ matrix~$R$ defines some basis of a $k$-dimensional vector subspace~$V_k$ of a vector space~$V_n$ and vice versa. The number of different bases in a given subspace~$V_k$ depends solely on the dimension~$k$, not on the contents of~$V_k$. On the other hand, the number of bases is equal to $|f^{-1}_R(V_k)|$, so the condition~\eqref{eq:PreimageCardinality} holds for the map~$f_R$.

Affine subspaces $f_t(t_1)$ and $f_t(t_2)$~\eqref{eq:f_t} for different $t_1, t_2 \in \F_2^n, t_1 \ne t_2$, either coincide or do not intersect. Indeed, suppose partial intersection and take $z \in f_t(t_1) \cap f_t(t_2)$, then $z = y_1 + t_1 = y_2 + t_2, \: y_1, y_2 \in V_k$. Consequently, $t_2 = t_1 + y_1 - y_2 = t_1 + \Delta y, \Delta y \in V_k$, and thus $f_t(t_2) \subset f_t(t_1)$. Analogously, one can prove that $f_t(t_1) \subset f_t(t_2)$. These two mutual inclusions mean that affine subspaces coincide, $f_t(t_1) = f_t(t_2)$, which contradicts the initial assumption.

It also follows from the above reasoning that two affine subspaces coincide, iff $t_2 = t_1 + y$, where $y \in V_k$. Therefore, for each $t_1 \in \F_2^n$ there exist $2^k$ vectors $t_2$ that result in the same affine subspace. So, cardinality $|f^{-1}_t(A_k)| = 2^k$ is the same for all $A_k$, and the condition~\eqref{eq:PreimageCardinality} is satisfied.

We have shown that affine subspaces~$A_k$, viewed as a shift of the \emph{fixed} vector subspace~$V_k$ by a random vector~$t$, are uniformly distributed. Because, as we proved earlier, vector subspaces are also sampled uniformly, the set of \emph{all} affine subspaces is sampled uniformly.

By pointing out that all necessary structures, namely, $\mathcal Q$, $L$, and~$A_k$, are uniformly distributed, we complete the proof.
\end{proof}

Theorem~\ref{thm:SampleSk} does not tell anything about how to sample matrices~$R$. In our implementation we use the simplest possible method. First, fill $n \times k$ matrix~$R$ with random bits, where $\prob(0) = \prob(1) = 1/2$, and compute matrix rank over $\F_2$. If $\rank R = k$, then stop, otherwise repeat the procedure. We have taken a routine for matrix rank calculation that requires $\mathcal O(n^2 k)$ operations.

\begin{theorem}\label{thm:SkCardinality}
The cardinality of $\mathcal S_k$ is
\begin{equation}
C(n, k) = 2^{n + \frac{k(k+1)}{2}} \binom{n}{k}_2, \label{eq:Cnk}
\end{equation}
where
\begin{equation}
\binom{n}{k}_2 = \prod_{j=0} ^{k-1} \frac{2^n - 2^j}{2^k - 2^j},
\end{equation}
is a 2-binomial (Gaussian) coefficient.
\end{theorem}

\begin{proof}
As in the proof of theorem~\ref{thm:SampleSk}, we can divide evaluation of~$C(n, k)$ by counting all distinct quadratic forms $\mathcal Q$~\eqref{eq:f_Q}, linear forms $L$~\eqref{eq:f_c}, affine subspaces $A_k$~~\eqref{eq:f_t}, and multiplying the results.

The whole set of matrices $Q \in \F_2^{k \times k}$ has cardinality $2^{k^2}$. But it is divided into groups of $2^{k(k-1)/2}$ matrices, where each group corresponds to the same quadratic form~$\mathcal Q$. Therefore, the total number of distinct forms is given by the ratio of these quantities and is equal to $C_\text{quad.} = 2^{k(k+1)/2}$.

There are $C_\text{lin.} = 2^k$ different possible linear forms over $\F_2^k$.

The total number of $k$-dimensional vectors subspaces~$V_k$ of an $n$-dimensional vector space over $\F_2$ is equal to a 2-binomial coefficient $\binom{n}{k}_2$~\cite{Goldman_SAM1970}. Each of these subspaces can be shifted in $2^k$ ways (by adding a vector $t \in V_k$, see proof of theorem~\ref{thm:SampleSk}) resulting in the same affine subspace~$A_k$. This gives $2^n/2^k$ different cosets $V_k + t$. Therefore, the total number of affine subspaces~$A_k$ is equal to product $C_\text{aff.} = 2^{n-k} \binom{n}{k}_2$.

By multiplying the numbers $C_\text{quad.}, C_\text{lin.}$, and $C_\text{aff.}$, we obtain expression~\eqref{eq:Cnk}.
\end{proof}

Using the $q$-binomial theorem~\cite{Goldman_SAM1970}, it is easy to check that, indeed, $\sum_{k = 0}^n C(n, k) = C(n)$ [see Eq.~\eqref{eq:C(n)}].

\begin{theorem}\label{thm:SampleS}
Choose integer $k$ randomly with probability $\prob(k) = C(n, k)/C(n)$, $0 \le k \le n$, and generate $\ket{\psi} \in \mathcal S_k$ according to the theorem~\ref{thm:SampleSk}. Then $\ket{\psi}$ is uniformly sampled from $\mathcal S$.
\end{theorem}

\begin{proof}
The index~$k$ determines the set~$\mathcal S_k$ to sample, hence, $\prob(\psi \in \mathcal S_k) = \prob(k)$. Conditional probability of sampling $\ket{\psi}$ from $\mathcal S_k$ is, $\prob(\psi \mid \psi \in \mathcal S_k) = 1/C(n, k)$, because the algorithm from the theorem~\ref{thm:SampleSk} produces uniformly distributed states. Then overall probability of obtaining the state $\psi \in \mathcal S$ is equal to
\begin{equation}
\prob(\psi) = \prob (\psi \mid \psi \in \mathcal S_k ) \prob(\psi \in \mathcal S_k)  = \frac{1}{C(n)}.
\end{equation}
Every state is produced with the same probability, therefore, the sampling is uniform.
\end{proof}

The overall complexity of the procedure from theorem~\ref{thm:SampleS} is $\mathcal O(2^n\poly(n))$, because there are $2^n$ elements in $\ket{\psi}$~\eqref{eq:GenStabForm2}, and each element evaluation requires no more than $\poly(n)$ operations. Actually, in our implementation $\poly(n) = \mathcal O(n^3)$, since the complexity is dominated by calculation of $\rank R$.

We provide the Python code for the explicit sampling of random stabilizer states, which is available at GitHub~\cite{RandStab_GitHub}.

\section{Median of means estimator\label{sec:MedianOfMeans}}
Median of means estimation~\cite{Jerrum_TCS1968} is an enhancement over an empirical mean estimator that is robust against the outlier corruption. Consider~$N$ samples $x_1, \dots, x_N$ of a random variable~$x$. The empirical mean~$\hat x$ is defined as
\begin{equation}
\hat x = \frac{1}{N} \sum_i x_i.
\end{equation}
According to the Chebyshev inequality~$\hat x$ deviates from the expectation $\E x$ by more than~$\epsilon$ with probability at most~$\delta$:
\begin{equation}
\prob(|\hat x - \E x| \ge \epsilon) \le \delta = \frac{\Var x}{N \epsilon^2},
\end{equation}
where $\Var x$ denotes the variance of~$x$. Therefore, the sampling complexity for a mean estimator is
\begin{equation}
N = \frac{\Var x}{\epsilon^2 \delta}. \label{eq:SampleComplexityMean}
\end{equation}
To calculate the median of means estimator $\hat x_\text{MM}$ one, first, splits $N$ samples into $K$ batches, each containing $\lfloor N/K \rfloor$ representatives, and evaluates empirical means $\hat x_k$ over the group number~$k$. Then, $\hat x_\text{MM}$ is defined as follows:
\begin{equation}
\hat x_\text{MM} = \mathrm{median}(\hat x_1, \dots, \hat x_K),
\end{equation}
This estimator is substantially more robust, since for the choice $K = \log 1/\delta$ the following inequality holds:
\begin{equation}
\prob(|\hat x_\text{MM} - \E x| \ge \epsilon) \le \delta = \exp \left(-\frac{N \epsilon^2}{4 \Var x} \right).
\end{equation}
The sample complexity for the median of means is
\begin{equation}
N = \frac{4 \Var x}{\epsilon^2} \log 1/\delta. \label{eq:SampleComplexityMedianMean}
\end{equation}
Note the appearance of $\log 1/\delta$ instead of $1/\delta$ compared to~\eqref{eq:SampleComplexityMean}.

It turns out, however, that in the asymptotic limit $N \to \infty$, the central limit theorem (CLT) holds, and empirical mean~$\hat x$ is distributed normally: $\hat x \sim \mathcal N(\E x, \Var(x)/N)$. One can calculate the probability of deviation:
\begin{gather}
\prob(|\hat x - \E x| \ge \epsilon) = \delta = 1 - \erf \left( \sqrt\frac{N \epsilon^2}{2 \Var x} \right), \nonumber\\ 
\delta \le \exp \left(-\frac{N \epsilon^2}{2 \Var x} \right),
\end{gather}
where $\erf(y)$ is the Gauss error function, which satisfies the inequality: $1-\erf(y) \le \exp(-y^2)$ for $y \ge 0$. By expressing~$N$, we obtain:
\begin{equation}
N \le \frac{2 \Var x}{\epsilon^2} \log 1/\delta.
\end{equation}
In this case the empirical mean is also robust because it contains the logarithmic dependence $\log 1/\delta$ similar to the one for the median of means estimator.

If the classical shadow~$\hat \rho$~\eqref{eq:Shadow} is substituted into~\eqref{eq:ohat}, then the estimator~$\hat o$ takes a form of the linear combination of random variables~$f_i$. Each frequency~$f_i$ can be viewed as an empirical mean of roughly $N/P$ single-shot measurements. Since in our experiment we worked in the overexposure regime $N/P \to \infty$, the CLT conditions are satisfied with high accuracy, and the above reasonings about the same performance of mean and median of means estimators become valid.

\section{Compensation of the Gouy phase\label{sec:GouyPhase}}

\begin{figure}[b]
	\centering
	\includegraphics[width=\linewidth]{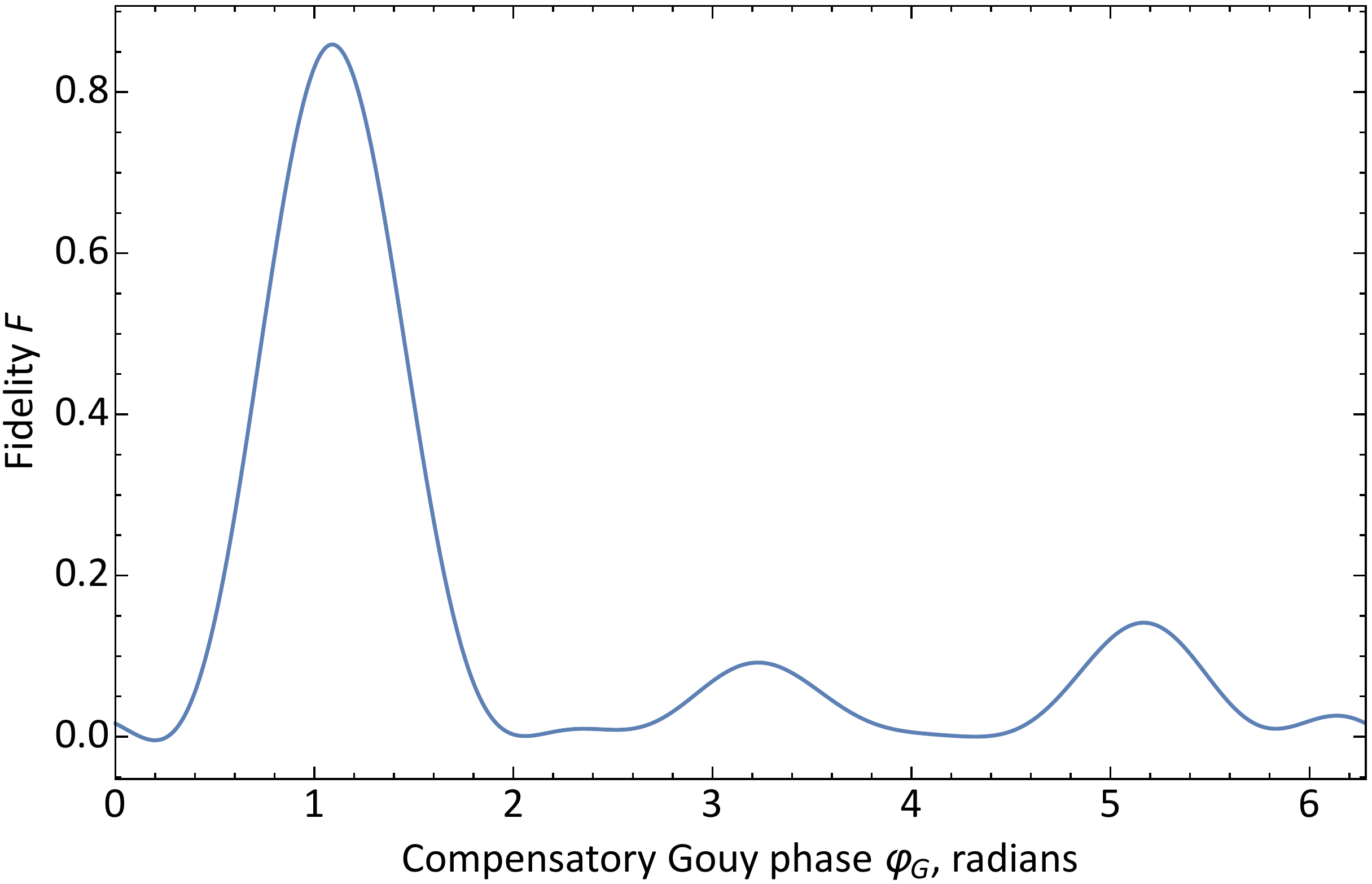}
	\caption{A typical dependence of preparation fidelity~$F$ on compensatory Gouy phase for $D = 32$. The global maximum corresponds to the true phase value.}
	\label{fig:GouyPhase}
\end{figure}

The prepared state~$\ket{\psi_\text{prep.}}$ and~the detected one~$\ket{\psi_\text{det.}}$ are tied by a unitary transformation~$U(\phi_G)$ with one unknown parameter, namely, Gouy phase~$\phi_G$: $\ket{\psi_\text{det.}} = U(\phi_G) \ket{\psi_\text{prep.}}$, where~$U(\phi_G)$ is a diagonal matrix and contains entities of unit magnitude only. Compensated fidelity of preparation~$F$ or \emph{preparation fidelity} for short is determined by maximizing fidelity to~$\ket{\psi_\text{det.}}$ over~$\phi_G$:
\begin{equation}
F = \max_{\phi_G} \me{\psi_\text{prep.}}{U^\dagger(\phi_G) \hat \rho U(\phi_G)}{\psi_\text{prep.}}. \label{eq:PreparationFidelity}
\end{equation}
The state $U(\phi_G^\text{max}) \ket{\psi^\text{prep.}}$, where $\phi_G^\text{max}$ maximizes~\eqref{eq:PreparationFidelity}, is the \emph{compensated prepared state}.

A typical dependence of the quantity under maximization in~\eqref{eq:PreparationFidelity} on~$\phi_G$ is demonstrated in Fig.~\ref{fig:GouyPhase} for the system dimensionality $D = 32$. There is a ``low-valued ripple'' with local extrema (especially pronounced for higher~$D$), but for all tested states and dimensions, a global maximum around $\phi_G \approx 1$ radians with near-unity value is observed. 

\bibliographystyle{apsrev4-2}
\bibliography{ref_base}

\end{document}